%% file: Formatting-Instructions-LaTeX-2026.tex
\documentclass[letterpaper]{article} 
\usepackage{aaai2026}  
\usepackage{times}  
\usepackage{helvet}  
\usepackage{courier}  
\usepackage[hyphens]{url}  
\usepackage{graphicx} 
\usepackage{natbib}  
\usepackage{caption} 
\usepackage{algorithm}

\usepackage{amssymb,amsmath,mathtools}
\usepackage{thm-restate}
\usepackage{nicefrac}
\usepackage[inline]{enumitem}
\usepackage{xspace}
\usepackage{amsthm}
\newtheorem{remark}{Remark}
\newtheorem{theorem}{Theorem}
\newtheorem{definition}{Definition}
\newtheorem{example}{Example}

\usepackage[capitalize]{cleveref}
\usepackage{nicefrac}
\usepackage{algpseudocodex}
\usepackage{booktabs}

\newcommand{\Nats}{\ensuremath{\mathbb{N}}\xspace}

\newcommand{\Reals}{\mathbb{R}}

\newcommand{\RR}{\Reals}
\newcommand{\NN}{\Nats}
\newcommand{\xx}{\boldsymbol{x}}
\newcommand{\uu}{\boldsymbol{u}}

\newcommand*\from{\colon}

\newcommand{\minimize}{\mathtt{minimize}}
\newcommand{\subj}{\mathtt{subj.~to}}

\newcommand{\ffor}{\ \ \textnormal{for}\ \ }

\newcommand{\defeq}{{}\triangleq{}}

\newcommand{\eeq}{~{}={}~}

\newcommand{\lleq}{~{}\leq{}~}

\newcommand{\ggeq}{~{}\geq{}~}

\def\triangleforqed{\hbox{$\lhd$}}
\makeatletter
\DeclareRobustCommand{\qedT}{%
	\ifmmode
	\eqno \def\@badmath{$$}
	\let\eqno\relax \let\leqno\relax \let\veqno\relax
	\hbox{\triangleforqed}%
	\else
	\leavevmode\unskip\penalty9999 \hbox{}\nobreak\hfill
	\quad\hbox{\triangleforqed}%
	\fi
}
\makeatother

\title{Runtime Safety and Reach-avoid Prediction of Stochastic Systems via Observation-aware Barrier Functions}
\author{
    Shenghua Feng\textsuperscript{\rm 1,2},
    Jie An\textsuperscript{\rm 1,2},
    Fanjiang Xu\textsuperscript{\rm 1,2}
}
\affiliations{
    \textsuperscript{\rm 1}National Key Laboratory of Space Integrated Information System, \\ Institute of Software Chinese Academy of Sciences, Beijing, China \\
    \textsuperscript{\rm 2}University of Chinese Academy of Sciences, Beijing, China \\
    \vspace{0.5em}
    \fontsize{10}{12}\selectfont
    \{fengshenghua,anjie,fanjiang\}@iscas.ac.cn
}

\begin{document}

\maketitle

\begin{abstract}
Stochastic dynamical systems have emerged as fundamental models across numerous application domains, providing powerful mathematical representations for capturing uncertain system behavior. In this paper, we address the problem of runtime safety and reach-avoid probability prediction for discrete-time stochastic systems with online observations, i.e., estimating the probability that the system satisfies a given safety or reach-avoid specification. Unlike traditional approaches that rely solely on offline models, we propose a framework that incorporates real-time observations to dynamically refine probability estimates for safety and reach-avoid events. By introducing observation-aware barrier functions, our method adaptively updates probability bounds as new observations are collected, combining efficient offline computation with online backward iteration. This approach enables rigorous and responsive prediction of safety and reach-avoid probabilities under uncertainty. In addition to the theoretical guarantees, experimental results on benchmark systems demonstrate the practical effectiveness of the proposed method.
\end{abstract}


\section{Introduction}

Stochastic dynamical systems provide robust mathematical frameworks for modeling real-world phenomena under uncertainty. These systems -- including Markov decision processes, probabilistic graphical models, and stochastic hybrid automata -- are pivotal in various fields, such as reinforcement learning, control theory, physics, signal processing, cryptography, finance, biology, and neuroscience~\cite{bertsekas1996stochastic, steele2001stochastic, allen2010introduction}. Ensuring reliability and safety in stochastic systems is a significant challenge, especially as these systems operate in increasingly complex and uncertain environments.

Safety and reach-avoid properties form the cornerstone of trustworthy stochastic system operations. Safety probability estimation quantifies the likelihood that a system trajectory remains outside unsafe regions during execution, whereas reach-avoid estimation assesses the probability of successfully reaching a target region without encountering unsafe states. These estimations are vital for safety-critical control, autonomous systems, robotics, and real-time decision-making~\cite{bertsekas1996stochastic,paul2013stochastic}, where precise risk assessments and robust guarantees are indispensable.

Traditional approaches for safety and reach-avoid predictions, such as stochastic barrier functions, rely on offline computations and predefined uncertainty models to establish probabilistic guarantees~\cite{prajna2004stochastic, feng2020unbounded, lechner2022stability, vzikelic2023learning}. These offline methodologies, however, fail to capitalize on real-time information like system observations or environmental dynamics, often leading to conservative or outdated predictions.

In this paper, we introduce a novel framework for runtime safety and reach-avoid prediction in discrete-time stochastic systems. Our method integrates real-time discrete observations gathered during system execution, enabling dynamic refinement of probability estimates. We utilize observation-aware barrier functions -- extensions of classical barrier certificates that adaptively update probability bounds in response to new observational data -- allowing our approach to rigorously and effectively reflect the evolving state of the system.
The proposed framework adopts a hybrid offline-online computational strategy. The offline phase involves the efficient synthesis of barrier functions through semidefinite programming techniques, addressing the computationally intensive aspects of the prediction process beforehand. This preparation greatly reduces the online computational burden. Subsequently, during system operation, the online phase leverages rapid backward iteration updates whenever new discrete observations become available. These updates dynamically recalibrate the barrier functions, swiftly refining the safety and reach-avoid probability estimates to align with the most current state information.

\paragraph{Contributions.} Our main contributions are as follows:
\begin{itemize}
\item Developing a framework for runtime safety and reach-avoid probability prediction using observation-aware barrier functions to incorporate online observations.
\item Providing theoretical guarantees on predicted probability bounds and introducing efficient runtime prediction algorithms.
\item Demonstrating significant improvements in adaptivity and accuracy over traditional methods through experiments on benchmark systems, underscoring the effectiveness of observation-aware verification in reliable stochastic control.
\end{itemize}

\section{Problem Formulation}

Let $\mathbb{R}$, $\mathbb{R}_{\geq 0}$, and $\mathbb{N}$ be the reals, nonnegative reals, and natural numbers, respectively. 
We consider a discrete-time stochastic dynamical system defined by:
\begin{equation}\label{eq:system}
\xx_{t+1} = f(\xx_t, \uu_t, \omega_t), \quad t \in \mathbb{N},
\end{equation}
where $\xx_t \in \mathcal{X} \subseteq \mathbb{R}^n$ denotes the state of the system at time $t$, $\uu_t \in \mathcal{U} \subseteq \mathbb{R}^m$ is the control input, and $\omega_t \in \mathcal{N} \subseteq \mathbb{R}^q$ represents a random disturbance. The disturbance sequence $\{\omega_t\}_{t \in \mathbb{N}}$ is assumed to be independent and identically distributed (i.i.d.) with a known probability distribution. The control input is determined by a policy $\pi: \mathcal{X} \to \mathcal{U}$, which is a measurable mapping assigning a state $\xx_t$ to the control $\uu_t = \pi(\xx_t)$. The transition function $f$ can be nonlinear, fully characterizing the system dynamics.

Given an initial state $\xx_0$ and a control policy $\pi$, a trajectory of the system is represented by the sequence $(\xx_0, \xx_1, \dots)$ satisfying $\xx_{t+1} = f\left(\xx_t, \pi(\xx_t), \omega_t\right)$, with $\omega_t$ drawn randomly according to its distribution. The resulting trajectory is denoted by $\{X_t^{\xx_0}\}_{t\in \mathbb{N}}$. This induces a probability measure $\mathbb{P}^{\pi}_{\xx_0}$ over all possible trajectories, which we will reference without superscripts or subscripts when the context is clear.

\paragraph{Safety and Reach-Avoid Probability Estimation.}
Consider the discrete-time stochastic system~\eqref{eq:system} with a given control policy $\pi$, an initial set $I \subseteq \mathbb{R}^n$, and an unsafe set $U \subseteq \mathbb{R}^n$. The \emph{safety probability estimation} problem seeks to determine a tight lower bound $\lambda$ on the probability that the system’s trajectory never enters the unsafe set $U$. Formally, the goal is to find a nontrivial $\lambda$ such that
\begin{equation}
\mathbb{P}_{\xx_0}\left(X_n \not\in U \text{ for all } n \in \mathbb{N}\right) \geq \lambda,
\end{equation}
for all initial states $\xx_0 \in I$.
In the \emph{reach-avoid} scenario, an additional target set $T \subseteq \mathbb{R}^n$ is specified. The objective is then to estimate a lower bound $\lambda$ on the probability that the system reaches the target set $T$ while avoiding the unsafe set $U$ up to that point. Specifically, the reach-avoid probability estimation problem is to find a nontrivial $\lambda$ such that
\begin{equation}
\mathbb{P}_{\xx_0}\left(\mathrm{RA(U,T)}\right) \geq \lambda,
\end{equation}
for all initial states $\xx_0 \in I$, where $\mathrm{RA}(U,T)$ denotes the event that the trajectory reaches $T$ before entering $U$:
\[
\mathrm{RA}(U,T) := \left\{\exists n \in \mathbb{N} : (X_n \in T) \wedge \big(\forall i < n : X_i \notin U\big)\right\}.
\]

Traditional methods addressing these problems frequently employ \emph{stochastic barrier functions} to establish safety and reach-avoid guarantees. However, these methods typically rely on offline computations based on prior system knowledge and statistical models of uncertainty, lacking adaptability to real-time observations. Such offline analyses fail to fully exploit valuable runtime data that could significantly refine safety assessments.

In practical scenarios, discrete-time observations become available at runtime, revealing that the system state resides within certain subsets $O_i$ at discrete time instances $t_i$. These runtime observations provide critical updates regarding system behavior and environmental interactions, potentially altering the probabilities associated with safety or reach-avoid conditions.
Consequently, traditional stochastic estimation methods must be adapted to dynamically integrate this observational data. Motivated by this challenge, we introduce the problem of \emph{runtime safety and reach-avoid prediction}, where real-time observations are leveraged to to yield more accurate and responsive safety predictions.

\paragraph{Runtime Safety and Reach-avoid Prediction.}
Consider again the stochastic system \eqref{eq:system} with control policy $\pi$, unsafe set $U$, initial set $I$, and a sequence of observation times $\{t_i\}_{i \in \mathbb{N}}$ with $t_0 = 0 < t_1 < t_2 < \dots$. At each observation time~$t_i$, the system state is observed to belong to a subset $O_i \subseteq \mathbb{R}^n$. The task is to dynamically compute or update the probability that, conditioned on the system being in $O_i$ at each observation time $t_i$, the trajectory remains safe (or satisfies the reach-avoid property) from that point onward. Formally, given any finite observation sequence $\{(t_i, O_i)\}_{i=1}^k$, the objective is to find a nontrivial $\lambda$ such that
\begin{equation}
\mathbb{P}_{\xx_0}\left(X_n \not\in U \text{ for all } n \in \Nats \mid  X_{t_i}\in  O_i \text{ for } i\leq k \right) \geq \lambda
\end{equation}
or, for the reach-avoid case:
\begin{equation}
\mathbb{P}_{\xx_0}\left(\,\mathrm{RA}(U,T) \mid  X_{t_i}\in  O_i \text{ for } i\leq k \,\right) \geq \lambda,
\end{equation}
where probability estimates are updated iteratively as new observations become available at each observation time.
\begin{remark}
For well-posedness, we assume the initial set $I$ and target set $T$ do not intersect the unsafe set $U$, and each observation $O_i$ is disjoint from both $U$ and $T$.
\end{remark}

\begin{remark}
As new observations are incorporated, the conditional safety or reach-avoid probability may increase or decrease, depending on observations. Therefore, the probability does not necessarily change monotonically, reflecting the adaptive and data-driven nature of our framework.
\end{remark}

Throughout the process of incorporating observations, we implicitly assume that the system has not entered the unsafe set $U$ (or reached the target set, in the reach-avoid setting) at any prior time. This assumption is both natural and practically justified, as entering the unsafe region or reaching the target is typically a detectable event in real systems. Once the system is observed to have entered $U$ or reached the target, further prediction or updating of safety or reach-avoid probabilities is no longer necessary. 

\section{Theoretical Results: Observation-aware Barrier Functions}

In this section, we develop a theoretical framework for estimating safety and reach-avoid probabilities conditioned on runtime observations. We show that runtime safety and reach-avoid predication can be systematically characterized using the proposed observation-aware barrier functions. We first address runtime safety estimation in detail, and then briefly discuss the reach-avoid case. Proof sketches for the main results are provided below, with complete proofs given in the appendix.



For runtime safety probability estimation, we leverage Bayes' theorem to reformulate the conditional safety probability in terms of joint and marginal probabilities involving both the system trajectory and the observation sequence. Given a finite sequence of observations $\{(t_i, O_i)\}_{i=1}^k$, the conditional safety probability can be written as
\begin{align*}
    ~&\mathbb{P}_{\xx_0}\left(X_n \not\in U \text{ for all } n \in \Nats \mid  X_{t_i}\in  O_i \text{ for } i\leq k \right) \\
    =~& 1 -  \mathbb{P}_{\xx_0}\left(\exists\,n,\,X_n \in U  \mid  X_{t_i}\in  O_i \text{ for } i\leq k \right)\\
    =~& 1 - \frac{
        \mathbb{P}_{\xx_0}\big(( \exists\,n,\,X_n \in U ) \wedge ( X_{t_i}\in  O_i \,, \forall i\leq k )\big)
    }{
        \mathbb{P}_{\xx_0}\left( X_{t_i}\in  O_i,\, \forall i\leq k \right)
    }.
\end{align*}
This formulation reveals that bounding the conditional safety probability reduces to estimating two terms: the probability of observing both a safety violation and the given observation sequence (the numerator), and the probability of observing the sequence alone (the denominator). If we can compute a lower bound $q$ for the denominator and an upper bound $p$ for the numerator, we immediately obtain
\[
    \mathbb{P}_{\xx_0}\left(X_n \not\in U \text{ for all } n \in \Nats \mid  X_{t_i}\in  O_i \text{ for } i\leq k \right) \geq 1 - \frac{p}{q} \,.
\]
This reduces the original problem to estimating two manageable probability bounds.


To facilitate this, we introduce \emph{observation-aware barrier functions} (OBFs) and \emph{observation-aware safety barrier functions} (OSBFs), which serve as core analytical tools for bounding the denominator and numerator, respectively, in the conditional probability formulation above. These functions extend classical barrier function techniques by explicitly incorporating information from discrete-time observations, enabling more accurate and adaptive safety prediction in the presence of runtime data.

\paragraph{Observation-aware Barrier Functions.} 
We define an \emph{observation-aware barrier function} (OBF) to be a function $B : \Nats_{\leq (t_k+1)} \times \mathcal{X} \to \mathbb{R}$, where $\Nats_{\leq (t_k+1)} := \{0, 1, \ldots, t_k+1\}$ denotes the discrete time indices from the initial time up to one step after the last observation. 
The OBF assigns values to time-state pairs and generalizes classical barrier functions to estimate the probability of observing a specified sequence of events, independent of safety requirements. The construction of OBFs enables lower bounds on observation probabilities to be propagated and updated as new information is gathered during system execution.

Intuitively, the value of an OBF is required to remain within $[0,1]$ everywhere (\emph{Probability condition}), to be lower bounded by $q$ at the initial state (\emph{Initial condition}), and to reach $1$ at the terminal time after the last observation (\emph{Terminal condition}). Along system trajectories, the OBF must not decrease in expectation before and at each observation time, with conditions further refined depending on whether the current state satisfies the observation at time $t$ (\emph{Safe before observation} and \emph{Observation-aware increase}). These properties ensure that, if the system consistently matches all observations while remaining safe, the OBF can only increase, thereby certifying a rigorous lower bound on the probability of realizing the entire observation sequence.



\begin{definition}[Observation barrier functions]\label{def:OBF}
Let $I \subseteq \mathbb{R}^n$ and $U \subseteq \mathbb{R}^n$ be the initial set and unsafe set, respectively, and let $q$ be the probability threshold. Given an observation sequence $\{(t_i, O_i)\}_{i=1}^k$, a function $B \colon \Nats_{\leq (t_k+1)}\times \mathcal{X} \to \mathbb{R}$ is said to be an observation barrier function (OBF) with respect to $I$, $U$, $\{(t_i, O_i)\}$, and $q$ if it satisfies:
\begin{enumerate}
\item Probability condition: $0 \leq B(t, \xx)\leq 1$ for all $(t, \xx) \in \Nats_{\leq (t_k +1)} \times \mathcal{X}$;
\item Initial condition: $B(0, \xx)\geq q$ for all $\xx\in I$;
\item Terminal condition: $B(t_k +1, x) = 1$ for all $x \in \mathcal{X}$;
\item Safe before observation: for all $t \in \{ 0, 1, \dots, t_k\}$,
$$ [\xx \in \mathcal{X}\setminus U]\cdot \mathbb{E}_{\omega \sim d} [ B(t+1, f(\xx, \pi(\xx), \omega_t)) ]\geq  B(t, \xx) \,.$$
\item Observation-aware increase: for all $t \in \{ t_1, t_2, \dots, t_{k}\}$,
$$ [\xx \in  O_t]\cdot \mathbb{E}_{\omega \sim d} [ B(t+1, f(\xx, \pi(\xx), \omega_t)) ]\geq  B(t, \xx) \,.$$
\end{enumerate}
\end{definition}
\begin{remark}
The \emph{safe before observation} condition in Definition~\ref{def:OBF} explicitly reflects the assumption that the system has not entered the unsafe set before each observation time. Consequently, all probabilities in this framework are conditioned on prior safety, aligning with practical scenarios where entering the unsafe region is immediately observable and thus further predictions become unnecessary.
\end{remark}


The following theorem formalizes the probabilistic guarantee provided by the existence of an observation-aware barrier function. Specifically, it asserts that if such a function can be constructed, then the probability of realizing the entire prescribed observation sequence is lower bounded by the threshold $q$.

\begin{restatable}[]{theorem}{obfobservationguarantee}
\label{thm:obf-observation-guarantee}
Suppose there exists an observation-aware barrier function $B(t, x)$ for the system with initial set $I$, unsafe set $U$, observation sequence $\{(t_i, O_i)\}_{i=1}^k$, and threshold $q$. Then, for any $x_0 \in I$,
\[
    \mathbb{P}_{x_0}\left(X_{t_i} \in O_i,~\forall i \leq k\right) \geq q.
\]
\end{restatable}

\begin{proof}[Proof sketch]
The main idea is to construct a submartingale based on the OBF, such that its expected value at time $t_k + 1$ exactly equals the probability of realizing the prescribed observation sequence.
Intuitively, a submartingale is a stochastic process whose expected value does not decrease over time; it captures the notion that, as the system evolves, the likelihood of satisfying the observation constraints cannot decrease unexpectedly. This property is crucial for establishing lower bounds on observation probabilities.
Specifically, for the process
\[
Y_n := \prod_{i < n} [X_i \in (\mathcal{X}\setminus U)] \prod_{t_i < n} [X_{t_i} \in O_{t_i}]\cdot  B(n, X_n),
\]
the OBF conditions guarantee that $Y_n$ is a submartingale. By the submartingale property and the OBF initial condition, we have $q \leq B(0, x_0) \leq \mathbb{E}[Y_{t_k+1}]$, which exactly equals the desired observation probability. This establishes the lower bound.
\end{proof}



\paragraph{Observation-aware Safety Barrier Functions}
While the observation-aware barrier function provides a rigorous lower bound on the probability of realizing a specified sequence of observations, runtime safety prediction also requires quantifying the probability of safety violations. To address this, we introduce the notion of an \emph{observation-aware safety barrier function} (OSBF).

\begin{definition}[Observation-aware safety barrier functions]
\label{def:OSBF}
Let $I \subseteq \mathbb{R}^n$ be the initial set, $U \subseteq \mathbb{R}^n$ the unsafe set, and let $p$ be the probability threshold. Given an observation sequence $\{(t_i, O_i)\}_{i=1}^k$, a function $V: \mathbb{N} \times \mathcal{X} \to \mathbb{R}$ is said to be an \emph{observation-aware safety barrier function} (OSBF) with respect to $U$, $I$, $\{(t_i, O_i)\}$, and $p$ if it satisfies:
\begin{enumerate}
    \item Nonnegativity: $V(t, \xx) \geq 0$ for all $(t, \xx) \in \mathbb{N} \times \mathcal{X}$. \label{cond:Nonnegativity}
    \item Initial condition: $V(0, \xx) \leq p$ for each $\xx \in I$. \label{cond:Initial}
    \item Safety condition: $V(t, \xx) \geq 1$ for $t \geq t_k + 1$ and $\xx \in U$.\label{cond:Safety}
    \item Safe before observation: for $t \in \{0, 1, \ldots, t_k\}$, \label{cond:Safe_before_observation}
    \[
        [\xx \in \mathcal{X}\setminus U] \cdot \mathbb{E}_{\omega \sim d}\big[ V(t+1, f(\xx, \pi(\xx), \omega)) \big] \leq V(t, \xx).
    \]
    \item Observation-aware decrease: for $t \in \{t_1, \ldots, t_k\}$, \label{cond:Observation-aware_decrease}
    \[
        [\xx \in  O_t] \cdot \mathbb{E}_{\omega \sim d}\big[ V(t+1, f(\xx, \pi(\xx), \omega)) \big] \leq V(t, \xx).
    \]
    \item Expected decrease after last observation: for $t \geq t_k+1$, and $\xx \in \mathcal{X}\setminus U$, \label{cond:decrease_after_last_observation}
    \[
        \mathbb{E}_{\omega \sim d}\big[ V(t+1, f(\xx, \pi(\xx), \omega)) \big] \leq V(t, \xx).
    \]
\end{enumerate}
\end{definition}
\begin{remark}
Compared to the OBF, the OSBF includes an additional condition---the expected decrease after the last observation (cond~\ref{cond:decrease_after_last_observation}). This condition is essential for certifying the system's safety beyond the observation horizon.
\end{remark}

Intuitively, an OSBF must be nonnegative (\emph{Nonnegativity}), bounded above at the initial state (\emph{Initial condition}), and bounded below on unsafe states after the last observation (\emph{Safety condition}). Its value must not increase in expectation before and at each observation time, with specific conditions depending on whether the state matches the most recent observation (\emph{Safe before observation} and \emph{Observation-aware decrease}). After the final observation, a standard expected decrease applies (\emph{Expected decrease}). Collectively, these properties ensure that if the system stays safe and satisfies all observations, the OSBF cannot increase unexpectedly, and any safety violation is flagged when the function exceeds a certain threshold. The following theorem formalizes this guarantee.

\begin{restatable}[]{theorem}{osbfsafetyguarantee}
    \label{thm:osbf-safety-guarantee}
    Suppose there exists an OSBF $V(t,\xx)$ for the system with initial set $I$, unsafe set $U$, observation sequence $\{(t_i, O_i)\}_{i=1}^k$, and threshold $p$. Then, for any $\xx_0 \in I$,
    \[
        \mathbb{P}_{\xx_0}\big(( \exists\,n,\,X_n \in U ) \wedge ( X_{t_i}\in  O_i \,, \forall i\leq k )\big) \leq p.
    \]
\end{restatable}

\begin{proof}[Proof sketch]
The complete proof is provided in the Appendix. The argument is similar to the previous result: we construct a supermartingale based on the OSBF, whose expected value at the stopping time for entering the unsafe set captures the safety violation probability. The result follows from the optional stopping theorem and the OSBF conditions.
\end{proof}

Together, the OSF and the OSBF provide a compositional approach to lower bounding the conditional safety probability. By separately certifying a bound for the joint probability of safety and observation events (via the OSBF) and a bound for the probability of the observation sequence (via the OBF), we can combine these results through the Bayesian reformulation established earlier. The following theorem formalizes this compositional guarantee:

\begin{restatable}[]{theorem}{conditionalsafety}
\label{thm:conditional-safety}
Suppose the notations and assumptions above hold, and there exist an OSBF $V(t,\xx)$ and an OBF $B(t,\xx)$ for the system with initial set $I$, unsafe set $U$, observation sequence $\{(t_i, O_i)\}_{i=1}^k$, and thresholds $p$ and $q$, respectively. Then, for any $\xx_0 \in I$,
\[
    \mathbb{P}_{\xx_0}\left(X_n \not\in U \text{ for all } n \in \Nats \mid X_{t_i} \in O_i \text{ for } i \leq k \right) \geq 1 - \frac{p}{q}.
\]
\end{restatable}

\paragraph{The Reach-avoid Case.}
The \emph{reach-avoid} scenario can be treated analogously to the safety case, with the additional assumption that the system will eventually enter $U \cup T$ with probability one. Under this assumption, the conditional reach-avoid probability can be expressed as
\begin{align*}
    ~&\mathbb{P}_{\xx_0}\left(\,\mathrm{RA}(U,T) \mid  X_{t_i}\in  O_i \text{ for } i\leq k \,\right)  \\
    =~& 1 - \mathbb{P}_{\xx_0}\left(\,\mathrm{RA}(T,U) \mid  X_{t_i}\in  O_i \text{ for } i\leq k \,\right) \\
    =~& 1 - \frac{
        \mathbb{P}_{\xx_0}\big( \mathrm{RA}(T,U) \wedge  ( X_{t_i}\in  O_i \,, \forall i\leq k )
        \big)
    }{
        \mathbb{P}_{\xx_0}\left( X_{t_i}\in  O_i,\, \forall i\leq k \right)
    }.
\end{align*}
The second equality holds because the system is guaranteed to enter either $U$ or $T$ eventually. Here, $\mathrm{RA}(T,U)$ denotes the event that the trajectory reaches $U$ before entering $T$:
\[
\mathrm{RA}(T,U) := \left\{ \exists\,n\,:\, X_n \in U,\,\, \forall j<n,\, X_j \not\in T \right\}.
\]
Therefore, the reach-avoid probability estimation reduces to obtaining a lower bound for the denominator and an upper bound for the numerator. Notably, the denominator coincides with that in the safety case, so only the numerator requires a new estimate.

\begin{remark}
Rather than directly lower bounding the conditional reach-avoid probability, we focus on upper bounding $\mathbb{P}_{\xx_0}\left(\mathrm{RA}(T,U) \mid  X_{t_i}\in  O_i,\,\forall i\leq k \right)$. This approach is preferable because upper bounds are typically easier to obtain, and the denominator estimate from the safety case can be reused.
\end{remark}

To upper bound the numerator above, we introduce the notion of an \emph{observation-aware reach-avoid barrier function} (ORBF), which generalizes the observation-aware safety barrier function to the reach-avoid setting.

\begin{definition}[Observation-aware reach-avoid barrier functions]
\label{def:ORBF}
Let $I \subseteq \mathbb{R}^n$ be the initial set, $U \subseteq \mathbb{R}^n$ the unsafe set, $T$ the target set, and let $p$ be the probability threshold. Given an observation sequence $\{(t_i, O_i)\}_{i=1}^k$, a function $V: \mathbb{N} \times \mathcal{X} \to \mathbb{R}$ is said to be an \emph{observation-aware reach-avoid barrier function} (ORBF) with respect to $U$, $T$, $I$, $\{(t_i, O_i)\}$, and $p$ if it satisfies conditions~\ref{cond:Nonnegativity}--\ref{cond:Observation-aware_decrease} in Definition~\ref{def:OSBF}, and further satisfies:
\begin{enumerate}
    \item[6.] Expected decrease after last observation (RA case): for $t \geq t_k+1$ and $\xx \in \mathcal{X} \setminus (U\cup T)$,\
    \[
        \mathbb{E}_{\omega \sim d}\big[ V(t+1, f(\xx, \pi(\xx), \omega)) \big] \leq V(t, \xx).
    \]
\end{enumerate}
\end{definition}

\begin{remark}
The ORBF differs from the OSBF in that its final expected decrease condition applies to all $\xx \in \mathcal{X} \setminus (U \cup T)$, reflecting that safety requirements cease once the target set $T$ is reached.
\end{remark}

The ORBF provides a rigorous characterization of the barrier property needed to upper bound the probability of violating the reach-avoid objective, given the observation constraints. The main probabilistic guarantees are formalized as follows:
\begin{restatable}[]{theorem}{ORBFguarantee}
\label{thm:ORBF_guarantee}
If there exists an ORBF $V(t, \xx)$ for the system with initial set $I$, unsafe set $U$, target set $T$, observation sequence $\{(t_i, O_i)\}_{i=1}^k$, and threshold $p$, then for any $\xx_0 \in I$,
\[
    \mathbb{P}_{\xx_0}\big( \mathrm{RA}(T,\,U) \wedge (X_{t_i} \in O_i,\ \forall i\leq k) \big) \leq p.
\]
\end{restatable}

Combining this result with the previously established lower bound on the probability of observing the sequence, we obtain the following conditional reach-avoid probability guarantee:

\begin{restatable}[]{theorem}{conditionalRA}
\label{thm:condtional-RA}
Suppose the notations and assumptions above hold, and the system enters $U \cup T$ with probability one. If there exist an ORBF $V(t, \xx)$ and an OSBF $B(t, \xx)$ for the system with initial set $I$, unsafe set $U$, target set $T$, observation sequence $\{(t_i, O_i)\}_{i=1}^k$, and thresholds $p$ and $q$ respectively, then for any $\xx_0 \in I$,
\[
    \mathbb{P}_{\xx_0}\big( \mathrm{RA}(U, T) \mid X_{t_i} \in O_i\,, \forall i\leq k \big) \geq 1 - \frac{p}{q}.
\]
\end{restatable}

\begin{remark}
The almost-sure reachability assumption -- that the system eventually enters $U \cup T$ with probability one -- can often be established using standard techniques such as stochastic ranking functions~\cite{chatterjee2016termination,chakarov2013probabilistic}. See also~\cite{majumdar2025sound} for related results on the termination of stochastic systems.
\end{remark}

\section{Algorithm for Runtime Safety and Reach-avoid Prediction}

In this section, we present practical algorithms for runtime safety and reach-avoid prediction based on observation-aware barrier functions. Our approach combines offline polynomial optimization and online backward iteration, enabling efficient updates of probability bounds upon receiving new observations.

The complete procedure is summarized in \cref{alg:main}. In the offline phase, we first synthesize a barrier function $v(\xx)$ satisfying the observation-independent conditions (conditions 1, 3, and 6 for OSBF and ORBF), which provides precomputed values for future time steps beyond the latest observation (lines 1--6). In the online phase, each time a new observation is received, we iteratively update the OBF, OSBF, or ORBF via backward computation (lines 7--17). Specifically, at each step, barrier values are updated using expectations over successor states, enforcing observation-aware and safety-related constraints. This backward iterative update dynamically refines the probability bounds in response to new runtime data, thus ensuring rigorous and adaptive runtime safety and reach-avoid prediction. The detailed computation procedures for OBF, OSBF, and ORBF are provided below.

\paragraph{Computing OBF.}
A tight OBF can be constructed via backward iteration as shown in \cref{alg:obf-backward}, procedure \textsc{Get-OBF}. Starting from the terminal time (where $B(t_k+1, x) = 1$ for all $x$), the OBF is recursively computed backward in time. At each step, the value is updated by the expected value over successor states, ensuring the safe-before-observation condition is satisfied. At observation times, the OBF is set to zero outside the observed set to enforce the observation-aware increase condition. This procedure ensures all conditions in \cref{def:OBF} are satisfied. Moreover, by maximizing the value at each step subject to these constraints, the method yields the largest possible OBF and, consequently, the tightest lower bound for the observation probability.

\begin{algorithm}[t!]
    \caption{Runtime Safety and Reach-avoid Prediction}
    \label{alg:main}
    \begin{algorithmic}[1]
    \footnotesize
    \Require System dynamics $f$, control policy $\pi$, initial set $I$, unsafe set $U$, (target set $T$ for reach-avoid), runtime observations
    \Ensure Probability lower bound for system safety or reach-avoid

    \LComment{Offline synthesis}
    \If{safety case}
    \State Synthesize $v(x)$ satisfying conds.~\eqref{cond:Nonnegativity}, \eqref{cond:Safety}, \eqref{cond:decrease_after_last_observation} in \cref{def:OSBF}
    \ElsIf{reach-avoid case }
    \State Verify the system enters $U \cup T$ almost surely
    \State Synthesize $v(x)$ satisfying conds.~\eqref{cond:Nonnegativity}, \eqref{cond:Safety}, (6) in \cref{def:ORBF}
    \EndIf

    \Statex
    \LComment{Online update}
    \Repeat
    \If{new observation $O_k$ at time $t_k$ }
    \LComment{Observation sequence: $\{(t_i, O_i)\}_{i=1}^k$}
    \State $q, B(t,x) \gets$ \Call{Get-OBF}{$\{(t_i, O_i)\}_{i=1}^k$}
    \If{safety case}
    \State $p, V(t,x) \gets$ \Call{Get-OSBF}{$\{(t_i, O_i)\}_{i=1}^k$,$v(\xx)$}
    \ElsIf{reach-avoid case }
    \State $p, V(t,x) \gets$ \Call{Get-ORBF}{$\{(t_i, O_i)\}_{i=1}^k$,$v(\xx)$}
    \EndIf
    \State Update bound: $1 - \nicefrac{p}{q}$
    \EndIf
    \Until{no new observation}
    \end{algorithmic}
\end{algorithm}

\paragraph{Computing OSBF and ORBF.}
For OSBF construction, we assume $V(t,\xx)$ is time-invariant after the last observation, i.e., $V(t,\xx) = v(\xx)$ for all $t \geq t_k + 1$. Under this assumption, the offline computation of $v(\xx)$ only needs to satisfy the nonnegativity, safety, and expected decrease constraints, independent of runtime observations. This corresponds to lines 1–3 in \cref{alg:main}. Specifically, the offline synthesis of $v(\xx)$ involves solving:
\begin{enumerate}
    \item $v(\xx) \geq 0$, for all $\xx \in \mathcal{X}$ (Nonnegativity);
    \item $v(\xx) \geq 1$, for all $\xx \in U$ (Safety condition);
    \item $\mathbb{E}_{\omega \sim d}[v(f(\xx,\pi(\xx),\omega))] \leq v(\xx)$, for all $\xx \in (\mathcal{X}\setminus U)$ (Expected decrease).
\end{enumerate}
When $f$, $\pi$, and $v$ are polynomials, these can be encoded as sum-of-squares (SOS) constraints and efficiently solved with semidefinite programming, e.g., using \textsc{Mosek}. See the appendix for formulation details.

Once $v(\xx)$ is synthesized, $V(t,\xx)$ for $t = 0, \dots, t_k$ is computed via backward iteration similar to OBF. At observation times, $V(t,\xx)$ is set to zero outside the observed set (to enforce the observation-aware decrease), while at other times, it is updated via the expectation over successor states, subject to the safe-before-observation condition.

The synthesis of ORBF follows the same structure as OSBF, with only the final expected decrease condition modified for the reach-avoid setting. Once $v(\xx)$ is obtained, the ORBF is constructed by backward iteration, as above.

\begin{algorithm}[t!]
\caption{Calculating OBF, OSBF, and ORBF}
\label{alg:obf-backward}
\begin{algorithmic}[1]
\footnotesize
\Require System dynamics $f$, policy $\pi$, initial set $I$, unsafe set $U$, (target set $T$ for reach-avoid).
\LComment{Backward calculation of OBF}
\Procedure{Get-OBF}{$\{(t_i, O_i)\}_{i=1}^k$}
\State $B(t_{k}+1, x) \gets 1$ for all $x\in \mathcal{X}$ \Comment{Terminal condition}
\For{$t = t_k$ \textbf{down to} $0$}
\If{$t$ is an observation time}
\State \(\begin{aligned}[t]
    & B(t, \xx) \\
    \gets &\begin{cases}
     \mathbb{E} [ B(t+1, f(\xx, \pi(\xx), \omega_t)) & \text{if }\xx\in O_t \\
     0 & \text{else}
\end{cases}
\end{aligned}\)
\Else
\State \(\begin{aligned}[t]
    & B(t, \xx) \\
    \gets &\begin{cases}
     \mathbb{E} [ B(t+1, f(\xx, \pi(\xx), \omega_t)) & \text{if }\xx\in \mathcal{X}\setminus U \\
     0 & \text{else}
\end{cases}
\end{aligned}\)
\EndIf
\EndFor
\State $q \gets \min_{\xx\in I} B(0, x)$ \Comment{Lower bound at initial set}
\State \Return $q, B$
\EndProcedure
\Statex

\LComment{Backward calculation of OSBF with offline $v(x)$}
\Procedure{Get-OSBF}{ $\{(t_i, O_i)\}_{i=1}^k$, $v(x)$ }
\State $V(t_{k}+1, x) \gets v(x) $ for all $x\in \mathcal{X}$ \Comment{Terminal condition}
\For{$t = t_k$ \textbf{down to} $0$}
\If{$t$ is an observation time}
\State \(\begin{aligned}[t]
    & V(t, \xx) \\
    \gets &\begin{cases}
     \mathbb{E} [V(t+1, f(\xx, \pi(\xx), \omega_t)) & \text{if }\xx\in O_t \\
     0 & \text{else}
\end{cases}
\end{aligned}\)
\Else
\State \(\begin{aligned}[t]
    & V(t, \xx) \\
    \gets &\begin{cases}
     \mathbb{E} [ V(t+1, f(\xx, \pi(\xx), \omega_t)) & \text{if }\xx\in \mathcal{X}\setminus U \\
     0 & \text{else}
\end{cases}
\end{aligned}\)
\EndIf
\EndFor
\State $p \gets \max_{\xx\in I} V(0, x)$ \Comment{Upper bound at initial set}
\State \Return $p, V$
\EndProcedure
\Statex
\LComment{ORBF shares the same procedure as OSBF, but with different $v(x)$}
\Procedure{Get-ORBF}{ $\{(t_i, O_i)\}_{i=1}^k$, $v(x)$ }
\State \Return \Call{Get-OSBF}{$\{(t_i, O_i)\}_{i=1}^k$, $v(x)$ }
\EndProcedure
\end{algorithmic}
\end{algorithm}

\begin{figure}
    \centering
    \begin{minipage}[t]{0.47\linewidth}
        \centering
        \includegraphics[width=\linewidth]{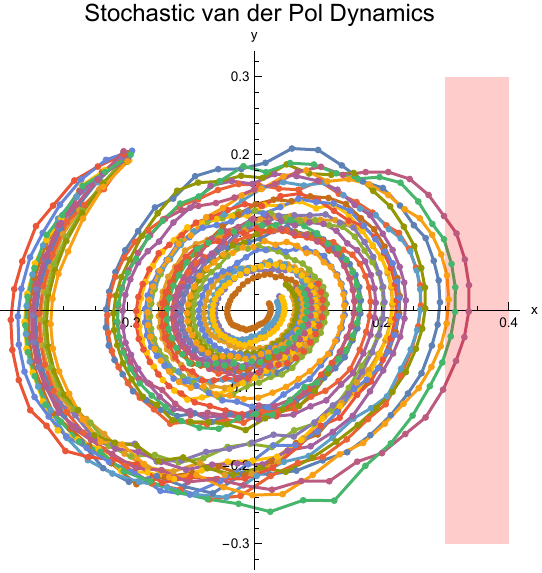}
    \end{minipage}
    \hfill
    \begin{minipage}[t]{0.47\linewidth}
        \centering
        \includegraphics[width=\linewidth]{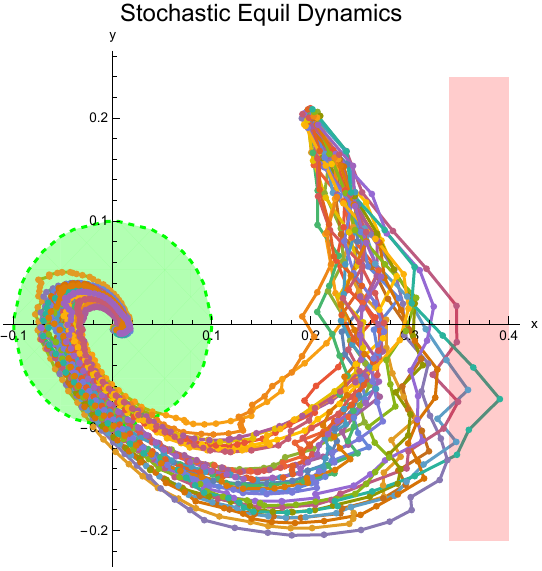}
    \end{minipage}
    \caption{Visualization of the stochastic Van der Pol system (safety) and Equil system (reach-avoid). The red region denotes the unsafe set; the green region denotes the target set.}
    \label{fig:vanderpol_equil}
\end{figure}

\begin{table*}[t]
	\begin{center}
		\begin{tabular}{l  rr c rr c rr c  rr c rr} 
			\toprule
			~ 
			&\multicolumn{2}{c}{No observation} 
                & ~
			&\multicolumn{2}{c}{2 observations}
                & ~
                &\multicolumn{2}{c}{3 observations}
                & ~
			&\multicolumn{2}{c}{4 observations}
                & ~
			&\multicolumn{2}{c}{5 observations}
			\\ \cmidrule{2-3} \cmidrule{5-6} \cmidrule{8-9} \cmidrule{11-12} \cmidrule{14-15} 
			Benchmark  
			& time (off.) &  prob. & ~ 
			& time (on.) & prob.  & ~  & time (on.) & prob.  & ~  & time (on.) & prob.  & ~ & time (on.) & prob. \\
			\midrule
                \textsf{arch} 
			& 2.39s  & 0.813  & ~
			& 0.01s  & 0.768   & ~
                & 0.03s & 0.780  & ~
                & 0.38s & 0.769  & ~
                & 0.51s & 0.574  
                \\
			\textsf{descent}  
			& 0.78s  & 0.704   & ~
			& 0.01s  & 0.626  & ~
                & 0.15s  & 0.655  & ~
                & 0.31s  & 0.502 & ~
                & 1.12s  & 0.124
                \\
			\textsf{osc}  
			& 1.49s & 0.436    & ~
			& 0.01s  & 0.532  & ~
                & 0.01s  & 0.470  & ~
                & 0.03s  & 0.850   & ~
                & 0.62s & 0.873  
                \\
                \textsf{vanderpol-1}  
			& 1.57s  & 0.343  & ~
			& 0.01s & 0.380  & ~
                & 0.06s  & 0.251  & ~
                & 0.17s  & 0.336  & ~
                & 0.23s  & 0.227  
                \\
                \textsf{vanderpol-2}  
			& 1.78s  & 0.158  & ~
			& 0.01s & 0.169  & ~
                & 0.06s  & 0.224 & ~
                & 0.15s  & 0.406  & ~
                & 0.21s  & 0.567  
                \\
			\textsf{liederivative}  
			& 6.80s  & 0.246   & ~
			& 0.01s  & 0.232 & ~
                & 0.15s  & 0.174  & ~
                & 1.26s  & 0.157  & ~
                & 1.76s & 0.056 
                \\
                \textsf{equil}  
			& 1.60s  & 0.357    & ~
			& 0.01s  & 0.274  & ~
                & 0.03s  & 0.239 & ~
                & 0.43s  & 0.194 & ~
                & 3.72s & 0.087  
                \\
			\textsf{lyapunov}  
			& 8.63s  & 0.226    & ~
			& 0.05s  & 0.283  & ~
                & 0.50s  & 0.279  & ~
                & 0.85s & 0.352  & ~
                & 2.50s & 0.308  
                \\
			\textsf{lotka}  
			& 14.96s  & 0.504     & ~
			& 0.02s  & 0.708  & ~
                & 0.42s  & 0.784  & ~
                & 0.61s & 0.907   & ~
                & 5.62s & 0.963 
                \\
			\bottomrule
		\end{tabular}
	\caption{
Experimental results for runtime safety and reach-avoid probability prediction. \textit{time (off)} denotes offline computation time without observations; \textit{time (on)} denotes online prediction time with observations; \textit{prob} is the probability lower bound for safety or reach-avoid.
}
	\label{tab:runtime_predication}
	\end{center}
\end{table*}

\section{Experiments}

To demonstrate the effectiveness and applicability of our runtime safety and reach-avoid prediction framework, we implemented the proposed algorithms in Python 3.13 (for backward iteration) and MATLAB R2025a interfaced with YALMIP~\cite{Lofberg2004} and MOSEK~\cite{DBLP:journals/mp/AndersenRT03} (for offline barrier synthesis). All experiments were conducted on a 2.60 GHz Intel Core i9-13905H laptop with 32 GB RAM, running 64-bit Windows 11.

\paragraph{Benchmarks and Experiment Setup.} 
We evaluated our framework on a set of polynomial benchmarks frequently used in the control literature. These include the Van der Pol oscillator~\cite{kanamaru2007van}, a classical nonlinear system commonly employed for verifying stochastic techniques due to its rich nonlinear dynamics, and the Equil system~\cite{prajna2007framework}, a variant of the Duffing oscillator extensively studied in control theory. The stochastic trajectories of these two representative benchmarks are illustrated in \cref{fig:vanderpol_equil}. Further details about each benchmark are provided in the appendix.
For each benchmark, we conducted runtime safety or reach-avoid predictions under different observation scenarios: no observation (purely offline prediction) and incremental online updates with up to five discrete-time observations.

\paragraph{Experimental Results and Analysis.}
The experimental results are summarized in in \cref{tab:runtime_predication}.
The experimental results clearly demonstrate the advantages and impact of incorporating runtime observations into safety and reach-avoid prediction:

\begin{itemize}
    \item \emph{Computation Efficiency}:  
    The initial offline safety or reach-avoid estimation, computed without runtime observations, typically incurs higher computational costs due to semidefinite optimization. In contrast, subsequent online predictions that incorporate runtime observations exhibit shorter computation times. This confirms that our online iterative updating scheme efficiently leverages the precomputed barriers, enabling rapid predictions suitable for real-time applications.

    \item \emph{Probability Refinement and Adaptation}:
    Runtime observations led to notable changes in the estimated safety and reach-avoid probabilities. As experiments show, additional runtime observations do not guarantee monotonically improving safety predictions. Depending on the observed states, predictions either enhanced the system's safety confidence (e.g., the osc benchmark) or, conversely, markedly reduced the estimated safety (e.g., descent benchmark). This reflects the adaptive nature of our method, as the predictions dynamically align with actual system trajectories and observations rather than solely relying on initial conservative estimates.

    \item \emph{Practical Implications}:
    The experiments underline the importance of runtime information: without observations, estimates remain conservative or overly optimistic, potentially misleading safety assessments. By incorporating discrete-time observations, the framework enables adaptive and realistic safety predictions, crucial for decision-making in safety-critical stochastic systems.
\end{itemize}

Overall, our experimental evaluations validate that the proposed runtime prediction framework effectively integrates online observations, efficiently computes updated predictions, and provides adaptively refined probability bounds essential for practical safety-critical applications.

\section{Related Work}

\paragraph{Verification of Deterministic Systems.}
The technique presented in this paper falls within the realm of barrier-based approaches.
In their seminal work~\cite{prajna2004stochastic}, Prajna proposed the concept of barrier certificates to encode inductive invariants that witness safety (or dually, reachability) of deterministic dynamical systems over an unbounded-time horizon. Since then, significant efforts have been dedicated to developing more relaxed forms of barrier-certificate conditions that still admit efficient synthesis, thereby leading to exponential-type barrier certificates~\cite{Kong13}, Darboux-type barrier certificates~\cite{zeng2016darboux}, general barrier certificates~\cite{Gan17}, and vector barrier certificates~\cite{Sogokon+Others/2018/Vector}. Similar barrier conditions have been utilized to verify systems with control inputs~\cite{xu2015robustness,ames2016control} and disturbances~\cite{wang2017generating} against various Linear Temporal Logic (LTL)~\cite{vardi2005automata} properties.

\paragraph{Verification of Stochastic Systems.}
Barrier-based methods have also been extended to verifying stochastic systems~\cite{DBLP:conf/cav/FengC00Z20,lechner2022stability,jagtap2020formal,vzikelic2023learning}.
There are also various alternative methods to reason about discrete-time stochastic dynamics, including techniques based on sampling~\cite{DBLP:conf/qest/HenriquesMZPC12,DBLP:conf/sefm/LegayST14}, dynamic programming~\cite{APLS08,DBLP:journals/automatica/SummersL10}, Markov abstractions~\cite{DBLP:journals/tac/LahijanianAB15}, probabilistic model checking~\cite{baier2008principles,DBLP:conf/lics/Kwiatkowska03} (for finite-state models), and various forms of value iteration~\cite{DBLP:conf/cav/Baier0L0W17,DBLP:conf/cav/HartmannsK20,DBLP:conf/cav/QuatmannK18} for determining reachability probabilities in Markov models.

\paragraph{Runtime Monitoring and Verification.} 
Runtime verification~\cite{bartocci2018introduction} is the process of dynamically monitoring a system during execution to ensure it adheres to specified safety or performance properties. Typically, properties are specified using temporal logic~\cite{deshmukh2017robust}, system signals are continuously monitored, and robustness metrics quantify how strongly a signal satisfies or violates the specification~\cite{raman2015reactive,su2025runtime}. Recently, learning-based runtime monitoring~\cite{yu2025neural,dawson2023safe} has attracted considerable attention since it allows scalable and adaptive monitoring in high-dimensional, uncertain environments.

\section{Conclusion}

We proposed a runtime prediction framework for safety and reach-avoid probabilities in discrete-time stochastic systems, effectively integrating real-time observations through observation-aware barrier functions. Our approach, combining offline computations with online updates, provides rigorous yet adaptive probability estimates, validated by experimental results on standard benchmarks. A promising avenue for future research is to leverage these runtime predictions for dynamically modifying control policies, thus enhancing real-time system safety and performance.

\section{Acknowledgments}
The authors would like to thank the anonymous reviewers for their valuable comments. This work is supported by the National Natural Science Foundation of China (Grants Nos. 62502475, W2511064, 62032024), and by the ISCAS Basic Research (Grant No. ISCAS-JCZD-202406). 

\bibliography{aaai2026}

\clearpage
\input{./appendix}

\end{document}

%% file: appendix.tex
\section{Appendix}

\subsection{Preliminaries on Martingale Theory}
This subsection introduces essential measure-theoretic preliminaries that underpin the proofs of the main results. For a more comprehensive introduction to probability and martingale theory, we refer the interested readers to~\cite{durrett2019probability, williams1991probability}.

A \emph{probability space} is a triple $(\Omega, \mathcal{F}, \mathbb{P})$, where $\Omega$ is a sample space, $\mathcal{F} \subseteq 2^\Omega$ is a $\sigma$-algebra on $\Omega$, and $\mathbb{P}\from \mathcal{F} \to [0, 1]$ is a probability measure on the measurable space $(\Omega, \mathcal{F})$. For any measurable space $(\Omega,\mathcal{F})$, denote the set of probability measure on $\Omega$ by $\mathcal{D}(\Omega)$.  A \emph{random variable} $X$ defined on the probability space $(\Omega, \mathcal{F}, P)$ is a $\mathcal{F}$-measurable function $X\from \Omega \to \RR\cup\{-\infty,+\infty\}$; its $\emph{expectation}$ (w.r.t. $\mathbb{P}$) is denoted by $E[X]$; For any set $A\subseteq \Omega$, $E[X\cdot[A]]$ is also denoted by $E[X;A]$. 

Let $\mathcal{F}'\subseteq \mathcal{F}$ is a sub-$\sigma$-algebra, a \emph{conditional expectation} of $X$ w.r.t. $\mathcal{F}'$ is a $\mathcal{F}'$-measurable random variable denoted by $E[X\mid \mathcal{F}']$, such that $E[X\cdot [A]]=E[E[X\mid \mathcal{F}']\cdot [A]]$ for all $A\in \mathcal{F}'$. A collection $\{\mathcal{F}_n \mid n \in \mathbb{N}\}$ of $\sigma$-algebras in $\mathcal{F}$ is a \emph{filtration} if $\mathcal{F}_n \subseteq \mathcal{F}_{n+k}$ for $n, k \in \mathbb{N}$. A random variable $T\from \Omega \to [0, \infty]$ is called a \emph{stopping time} w.r.t. some filtration $\{\mathcal{F}_n \mid n \in \mathbb{N}_0\}$ of $\mathcal{F}$ if $\{T \le n\} \in \mathcal{F}_n$ for all $n \in \mathbb{N}$.

\paragraph{Martingales.} A stochastic process $\{X_n\}_{n\in \mathbb{N}}$ adapted to a filtration $\{\mathcal{F}_n \mid n \in \mathbb{N}\}$ is called a \emph{supermartingale} (resp. \emph{submartingale}) if $E[X_n] < \infty$ for any $n \in \mathbb{N}_0$ and $E[X_{m} \mid \mathcal{F}_n] \le X_n$ (resp. $E[X_{m} \mid \mathcal{F}_n] \geq X_n$) for all $m \le n$. That is, the conditional expected value of any future observation, given all past observations, is no larger (resp. smaller) than the most recent observation. $\{X_n\}_{n\in \mathbb{N}}$ is a martingale if it is both supermartingale and submartingale.

Intuitively, A supermartingale is a stochastic process where, at any given time, the expected value of the next step is less than or equal to the current value, capturing the idea of a process that, on average, does not increase. In contrast, a submartingale is a process whose expected future value is at least as large as its present value, reflecting a tendency to increase over time.

The optional stopping theorem asserts that, under mild and natural conditions, no strategic choice of a stopping time can increase the expected value of a supermartingale (or decrease it for a submartingale). This formalizes the intuition that, in fair stochastic processes, timing alone cannot yield an expected advantage.

\begin{theorem}[Optional Stopping Theorem \cite{durrett2019probability,williams1991probability}]\label{thm:ost}
Let $T$ be a stopping time w.r.t. $\mathcal{F}_n$, and $\{X_n\}_{n\in \NN}$ is a supermartingale (resp. submaritngale) adapted to $\mathcal{F}_n$, such that $E[X_n]<\infty$ for all $n\in \NN$. Assume that one of the following three conditions holds:
\begin{itemize}
	\item $T$ is bounded almost surely, i.e. there exists $N\in\mathbb{N}$ such that $\mathbb{P}(T\leq N)=1$;
	\item $X_{n\wedge T}$ is bounded, i.e. there exists constant $C\in \RR^+$ such that $|X_{n\wedge T}|\leq C$ almost surely;
	\item $E[T]< \infty$ and $X_{n\wedge T}$ is conditional difference bounded, i.e. there exists $M>0$ such that $$E[|X_{(n+1)\wedge T}-X_{n\wedge T}|\mid \mathcal{F}_n]\leq M$$
\end{itemize}
then $E[X_T] \leq E[X_0]$ if $\{X_n\}_{n\in \NN}$ is a supermartingale. (resp. $E[X_T] \geq E[X_0]$ for the submartingale case).
\end{theorem}

\medskip
\subsection{Proof of \cref{thm:obf-observation-guarantee}}
\obfobservationguarantee*
\begin{proof}
The main idea is to construct a submartingale based on the OBF, such that its expected value at time $t_k + 1$ exactly equals the probability of realizing the prescribed observation sequence.
Specifically, let $Y_n$ be the process
\[
Y_n := \prod_{i < n} [X_i \in (\mathcal{X}\setminus U)] \prod_{t_i < n} [X_{t_i} \in O_{t_i}]\cdot  B(n, X_n),
\]
where $[\cdot]$ denotes the indicator function. According to the definition of OBF, for $t\leq t_k$, we have

\begin{equation*}
\begin{aligned}
 & E[Y_{n+1} \mid X_n] 
= 
\begin{aligned}[t]
& \prod_{i < n+1} [X_i \in (\mathcal{X}\setminus U)]  \prod_{t_i < n+1} [X_{t_i} \in O_{t_i}] \cdot \\
& E[B(n+1, X_{n+1}) \mid X_n]
\end{aligned}\\
= &\, \begin{aligned}[t]
& \prod_{i < n} [X_i \in (\mathcal{X}\setminus U)]  \prod_{t_i < n} [X_{t_i} \in O_{t_i}] \cdot \\
& [X_n \in (\mathcal{X}\setminus U)] [X_{t_n} \in O_{t_n}] E[B(n+1, X_{n+1}) \mid X_n]
\end{aligned}\\
\geq  & \, \prod_{i < n} [X_i \in (\mathcal{X}\setminus U)] \prod_{t_i < n} [X_{t_i} \in O_{t_i}]\cdot  B(n, X_n)\\
= &\, Y_n\,,
\end{aligned}
\end{equation*}
where the inequality follows from the conditions of the OBF. This implies $Y_n$ is indeed a submartingale for $t \leq t_k +1$. 
By the submartingale property, the expected value $\mathbb{E}[Y_{t_k+1}]$ is lower bounded by its initial value, i.e.
\[ \mathbb{E}[Y_{t_k+1}] \geq E[Y_0] \geq q ,
\]
Moreover, the terminal condition (3) implies 
\[ 
\begin{aligned}
    E[Y_{t_k+1}] = & E[\prod_{i \leq t_k } [X_i \in (\mathcal{X}\setminus U)] \prod_{t_i \leq t_k} [X_{t_i} \in O_{t_i}] ]\\
    = & \mathbb{P}_{x_0}\left(X_{t_i} \in O_i,~\forall i \leq k\right) 
\end{aligned}
\]
Note the second equality holds since we implicitly assume that the system has not entered the unsafe set before each observation time, and all probabilities in this paper are conditioned on prior safety, as in remark 3.
This completes the argument and establishes the desired lower bound.
\end{proof}

\medskip
\subsection{Proof of \cref{thm:osbf-safety-guarantee}}
\osbfsafetyguarantee*
\begin{proof}
We will construct a supermartingale based on the OSBF, whose expected value at the stopping time for entering the unsafe set captures the safety violation probability. Formally, let
$Y_n = \prod_{i < n} [X_i \in (\mathcal{X}\setminus U)] \prod_{t_i < n} [X_{t_i} \in O_{t_i}]\cdot  B(n, X_n),$
and let 
$T := \inf \{n \mid X_n \in U\}$. Clearly, $T\geq n+1$ when $Y_n \neq 0$ and $n \leq t_k$. Therefore, by the definition of OSBF, for $n\leq t_k$, we have
\begin{equation*}
\begin{aligned}
 & E[Y_{(n+1)\wedge T} \mid X_n] = E[Y_{(n+1)}\mid X_n] \\
= &\,
\begin{aligned}[t]
& \prod_{i < n+1} [X_i \in (\mathcal{X}\setminus U)]  \prod_{t_i < n+1} [X_{t_i} \in O_{t_i}] \cdot \\
& E[B(n+1, X_{n+1}) \mid X_n]
\end{aligned}\\
= &\, \begin{aligned}[t]
& \prod_{i < n} [X_i \in (\mathcal{X}\setminus U)]  \prod_{t_i < n} [X_{t_i} \in O_{t_i}] \cdot \\
& [X_n \in (\mathcal{X}\setminus U)] [X_{t_n} \in O_{t_n}] E[B(n+1, X_{n+1}) \mid X_n]
\end{aligned}\\
\leq  & \, \prod_{i < n} [X_i \in (\mathcal{X}\setminus U)] \prod_{t_i < n} [X_{t_i} \in O_{t_i}]\cdot  B(n, X_n)\\
= &\, Y_n = Y_{n\wedge T}\,,
\end{aligned}
\end{equation*}
For $n> t_k$, if $T \leq n$, we have $Y_{{n+1}\wedge T} =  Y_{n\wedge T} = Y_T$, this implies $ E[Y_{(n+1)\wedge T} \mid X_n] = E[Y_{n\wedge T}]$. If $T > n$, by condition \ref{cond:decrease_after_last_observation} in \cref{def:OSBF}, we have
\begin{equation*}
\begin{aligned}
 & E[Y_{(n+1)\wedge T} \mid X_n] = E[Y_{(n+1)}\mid X_n] \\
= &\,
\begin{aligned}[t]
& \prod_{i < n+1} [X_i \in (\mathcal{X}\setminus U)]  \prod_{i=1}^k [X_{t_i} \in O_{t_i}] \cdot \\
& E[B(n+1, X_{n+1}) \mid X_n]
\end{aligned}\\
= &\, \begin{aligned}[t]
& \prod_{i < n} [X_i \in (\mathcal{X}\setminus U)]  \prod_{i = 1}^k [X_{t_i} \in O_{t_i}] \cdot \\
& [X_n \in (\mathcal{X}\setminus U)] E[B(n+1, X_{n+1}) \mid X_n]
\end{aligned}\\
\leq  & \, \prod_{i < n} [X_i \in (\mathcal{X}\setminus U)] \prod_{i=1}^k [X_{t_i} \in O_{t_i}]\cdot  B(n, X_n)\\
= &\, Y_n = Y_{n\wedge T}\,,
\end{aligned}
\end{equation*}
Combining all together, we have $\{Y_{n\wedge T}\}_{n\in \Nats}$ is a supermartingale. By optional stopping theorem and cond~\ref{cond:Nonnegativity}, \ref{cond:Initial} and \ref{cond:Safety} in OSBF, we have
\[
\begin{aligned}
    &\,\mathbb{P}_{\xx_0}\big(( \exists\,n,\,X_n \in U ) \wedge ( X_{t_i}\in  O_i \,, \forall i\leq k )\big)\\
    \leq& \, E[Y_T] \leq E[Y_0] \leq p.
\end{aligned}
\]
This completes the proof and establishes the desired upper bound.
\end{proof}

\medskip
\subsection{Proof of \cref{thm:conditional-safety}}
\conditionalsafety*
\begin{proof}
The result follows directly from the following equality established before: 
    \begin{align*}
    ~&\mathbb{P}_{\xx_0}\left(X_n \not\in U \text{ for all } n \in \Nats \mid  X_{t_i}\in  O_i \text{ for } i\leq k \right) \\
    =~& 1 -  \mathbb{P}_{\xx_0}\left(\exists\,n,\,X_n \in U  \mid  X_{t_i}\in  O_i \text{ for } i\leq k \right)\\
    =~& 1 - \frac{
        \mathbb{P}_{\xx_0}\big(( \exists\,n,\,X_n \in U ) \wedge ( X_{t_i}\in  O_i \,, \forall i\leq k )\big)
    }{
        \mathbb{P}_{\xx_0}\left( X_{t_i}\in  O_i,\, \forall i\leq k \right)
    }.
\end{align*}
This complete the proof.
\end{proof}

\medskip
\subsection{Proof of \cref{thm:ORBF_guarantee}}
\ORBFguarantee*
\begin{proof}
The proof is analogous to the proof of \cref{thm:osbf-safety-guarantee}. We will construct a supermartingale based on the ORBF, whose expected value at the stopping time for entering the unsafe and target set captures the $\mathrm{RA}(T,U)$ probability.
Formally, Let
$Y_n = \prod_{i < n} [X_i \in (\mathcal{X}\setminus U)] \prod_{t_i < n} [X_{t_i} \in O_{t_i}]\cdot  B(n, X_n),$
and let $T := \inf \{n \mid X_n \in U\cup T\}$, stopping time $T$ represents the first time the system enters unsafe or targe set. 

Follow the same argument as in proof of \cref{thm:osbf-safety-guarantee}, we have $\{Y_{n\wedge T}\}_{n\in \Nats}$ is a supermartingale.  By optional stopping theorem, we have
\[
\begin{aligned}
    &\,\mathbb{P}_{\xx_0}\big( \mathrm{RA}(T,\,U) \wedge (X_{t_i} \in O_i,\ \forall i\leq k) \big)\\
    \leq& \, E[Y_T] \leq E[Y_0] \leq p\,,
\end{aligned}
\]
where the first inequality holds because $B(n,X_n)$ is required to be positive over $T$, and greater than $1$ over $U$, as in cond~\ref{cond:Nonnegativity} and \ref{cond:Safety} in ORBF. This completes the proof and establishes the desired upper bound.
\end{proof}

\medskip
\subsection{Proof of \cref{thm:condtional-RA}}
\conditionalRA*
\begin{proof}
Since the system is guaranteed to enter either $U$ or $T$ eventually, the result follows directly from the following equality established before: 
\begin{align*}
    ~&\mathbb{P}_{\xx_0}\left(\,\mathrm{RA}(U,T) \mid  X_{t_i}\in  O_i \text{ for } i\leq k \,\right)  \\
    =~& 1 - \mathbb{P}_{\xx_0}\left(\,\mathrm{RA}(T,U) \mid  X_{t_i}\in  O_i \text{ for } i\leq k \,\right) \\
    =~& 1 - \frac{
        \mathbb{P}_{\xx_0}\big( \mathrm{RA}(T,U) \wedge  ( X_{t_i}\in  O_i \,, \forall i\leq k )
        \big)
    }{
        \mathbb{P}_{\xx_0}\left( X_{t_i}\in  O_i,\, \forall i\leq k \right)
    }.
\end{align*}
This complete the proof.
\end{proof}

\medskip
\subsection{Details on Computing OSBF and ORBF}
We show in this subsection how to encode the synthesis of $v(\xx)$ as sum-of-squares (SOS) programming problems~\cite{parrilo2003semidefinite}. SOS programming refers to convex programs with linear objectives and sum-of-squares-shaped constraints; they can be translated to semidefinite programming (SDP) problems~\cite{vandenberghe1996semidefinite} that admit polynomial-time algorithms implemented by many off-the-shelf SDP solvers. 
The SOS formulation of $v(\xx)$ relies on the following assumptions: the flow map $f(x,\pi(\xx),\theta)$ is polynomial in $x$; the sets $\mathcal{X}$, $I$, $U$, and $T$ are all semi-algebraic, i.e., they can all be translated into the form $\{\xx \mid \bigvee_i\bigwedge_j P_{ij}(\xx) \mathrel{\triangleright} 0\}$ with polynomials $P_{ij}$ and $\triangleright \in \{{\geq}, {>}\}$.

We start by creating a \emph{polynomial template} $v^{a}(\xx)$ in $\xx$ of certain degree $d$ with unknown parameters $a$ (encoding the vector of unknown coefficients). The synthesis of $v(\xx)$ amounts to finding an appropriate valuation of $a$ such that the following barrier conditions are fulfilled:
\begin{align*}
	\underset{\gamma}{\minimize}\quad &\gamma\,; \quad \text{(Safety case)} \\
	\subj\quad &v^a(\xx) \lleq \gamma, \quad \text{for } \xx \in I\,,\\
	&v^a(\xx) \ggeq 0, \quad \text{for } \xx \in \mathcal{X}\,,\\
        &v^a(\xx) \ggeq 1, \quad \text{for } \xx \in U\,,\\
	&E\left[v^a(f(\xx,\pi(\xx),\theta))\right] \leq v^a(\xx) , \; \text{for } \xx \in \mathcal{X}\setminus U .
\end{align*}%
or in reach-avoid case:
\begin{align*}
	\underset{\gamma}{\minimize}\quad &\gamma\,; \quad \text{(RA case)} \\
	\subj\quad &v^a(\xx) \lleq \gamma, \quad \text{for } \xx \in I\,,\\
	&v^a(\xx) \ggeq 0, \quad \text{for } \xx \in \mathcal{X}\,,\\
        &v^a(\xx) \ggeq 1, \quad \text{for } \xx \in U\,,\\
	E[v^a(&f(\xx,\pi(\xx),\theta))] \leq v^a(\xx) , \; \text{for } \xx \in \mathcal{X}\setminus (U \cup T) .
\end{align*}%

Observe that all the constraints above share a common form
\begin{align*}\label{eq:sdp_example}
	V^a(\xx) \ggeq 0 \ffor  \xx \,\in\, \left\{\xx \mathrel{\Big|} \bigvee\nolimits_{i = 0}^m\bigwedge\nolimits_{j = 0}^{l} P_{ij}(\xx) \mathrel{\triangleright} 0\right\},
\end{align*}%
i.e., $V^a(\cdot)$ is a parametrized polynomial that is non-negative over a semi-algebraic set. Based on the well-known Putinar's Positivstellensatz~\cite{putinar1993positive}, the above constraint can be reformulated into a group of SOS constraints:
\begin{align*}
	&V^a(\xx) + \sum_{j=0}^{l} s_{ij}(x)\cdot P_{ij}(x) ~\in~ \textrm{sos}[x], \quad \text{for } 0\leq i \leq m\,,\\
	&s_{ij} ~\in~ \textrm{sos}[x], \quad \text{for } 0\leq i \leq m,\, 0\leq j \leq l
\end{align*}%
where $\textrm{sos}[x] \defeq \{g(x)\in \Reals[x] \mid g= h_1^2 + h_2^2+ \dots + h_k^2\}$ denotes the set of all sum-of-squares polynomials in $x$ (over the reals). Consequently, the constraints on 
$v(x)$ can be encoded as SOS programming problems, which can then be solved using an off-the-shelf SOS/SDP solver.

\medskip
\subsection{Benchmarks in Experiment}
The stochastic term $\theta$ in the following benchmarks all obey the uniform distribution over interval $[-1,1]$. 

\begin{example}[\textnormal{\textsf{vanderpol1}~\cite{XZF22}}]
The system dynamic is: (safety, with dynamics illustrated in \cref{fig:vanderpol_equil})
    \begin{align*}
        x_{n+1} & \eeq x_n - 0.2 y_n \\
        y_{n+1} & \eeq y_n +0.2(x_n+0.5 y_n (x_n^2-1-1.7\theta))
    \end{align*}
    \begin{itemize}
        \item  $\mathcal{X} =~ \{(x,y)\in\RR^2 \mid x^2 + y^2 \le 0.5^2\}$.
        \item $U =~ \{(x,y)\in\RR^2 \mid 0.3-x \le 0\}$.
        \item  $I ~=~  \{(x,y)\in\RR^2 \mid (x+0.2)^2 + (y-0.2)^2 \le 0.05^2\}$.
    \end{itemize} 
\end{example}

\smallskip

\begin{example}[\textnormal{\textsf{vanderpol2}~\cite{XZF22}}]
The system dynamic is: (safety)
    \begin{align*}
        x_{n+1} & \eeq x_n - 0.2 y_n \\
        y_{n+1} & \eeq y_n +0.2(x_n+0.5 y_n (x_n^2-1-1.7\theta))
    \end{align*}
    \begin{itemize}
        \item  $\mathcal{X} =~ \{(x,y)\in\RR^2 \mid x^2 + y^2 \le 0.5^2\}$.
        \item $U =~ \{(x,y)\in\RR^2 \mid 0.4-x \le 0\}$.
        \item  $I ~=~  \{(x,y)\in\RR^2 \mid (x+0.25)^2 + (y-0.25)^2 \le 0.01^2\}$.
    \end{itemize}  
\end{example}

\smallskip

\begin{example}[\textnormal{\textsf{equil}~\cite{PJP07}}]
The system dynamic is: (reach-avoid, with dynamics illustrated in \cref{fig:vanderpol_equil})
    \begin{align*}
        x_{n+1} & \eeq x_{n} + 0.1(y_n+ \theta x_n)\\
        y_{n+1} & \eeq y_n   + 0.1(-x_n+\nicefrac{x_n^3}{3} - y_n)
    \end{align*}
    \begin{itemize}
        \item  $\mathcal{X} =~ \{(x,y)\in\RR^2 \mid x^2 +  y^2 \le 0.5^2\}$.
        \item $U =~ \{(x,y)\in\RR^2 \mid 0.34-x\le 0\}$.
        \item $T =~ \{(x,y)\in\RR^2 \mid x^2 + y^2 \le 0.1^2\}$.
        \item  $I ~=~  \{(x,y)\in\RR^2 \mid (x-0.2)^2 + (y-0.2)^2 \le 0.01^2\}$.
    \end{itemize}  
\end{example}

\smallskip

\begin{example}[\textnormal{\textsf{arch}~\cite{SGJ16}}]
The system dynamic is: (safety)
    \begin{align*}
        x_{n+1} & \eeq x_n + 0.1(x_n - x_n^3 + \theta y_n - x_ny_n^2)\\
        y_{n+1} & \eeq y_n + 0.1(-x_n + \theta y_n - x_n^2y_n -y_n^3)
    \end{align*}
    \begin{itemize}
        \item  $\mathcal{X} =~ \{(x,y)\in\RR^2 \mid x^2 + y^2 \le 1\}$.
        \item $U =~ \{(x,y)\in\RR^2 \mid x^2 + y^2 \le 0.04\}$.
        \item  $I ~=~  \{(x,y)\in\RR^2 \mid (x-1)^2 + (y-1)^2 \le 0.04\}$.
    \end{itemize} 
\end{example}

\smallskip

\begin{example}[\textnormal{\textsf{descent}}]
The system dynamic is: (reach-avoid)
    \begin{align*}
        x_{n+1} \eeq x_{n} + 0.2(\theta -1).
    \end{align*}
    \begin{itemize}
        \item  $\mathcal{X} =~ \{(x,y)\in\RR^2 \mid (x-1)^2  \le 16\}$.
        \item $U =~ \{(x,y)\in\RR^2 \mid (x-3)^2 \le 0.01\}$.
        \item $T =~ \{(x,y)\in\RR^2 \mid (x-0.2)^2 \le 0.01\} $
        \item  $I ~=~  \{(x,y)\in\RR^2 \mid (x-1)^2 \le 0.04\}$.
    \end{itemize} 
\end{example}

\smallskip

\begin{example}[\textnormal{\textsf{osc}~\cite{XZF22}}]
The system dynamic is: (safety)
    \begin{align*}
        x_{n+1} & \eeq x_{n} + 0.1 y_n\\
        y_{n+1} & \eeq y_n   + 0.1(-x_n - (1.6-2\theta)y_n)
    \end{align*}
    \begin{itemize}
        \item  $\mathcal{X} =~ \{(x,y)\in\RR^2 \mid x^2 + y^2 \le 16\}$.
        \item $U =~ \{(x,y)\in\RR^2 \mid x^2 + (y-2)^2 \le 1\}$.
        \item  $I ~=~  \{(x,y)\in\RR^2 \mid x^2 + (y-0.75)^2 \le 0.01\}$.
    \end{itemize} 
\end{example}

\smallskip

\begin{example}[\textnormal{\textsf{liederivative}~\cite{LZZ11}}]
The system dynamic is: (safety)
    \begin{align*}
        x_{n+1} & \eeq x_n - 0.2y_n\\
        y_{n+1} & \eeq y_n + 0.1(x_n^2 + \theta x_n)
    \end{align*}
    \begin{itemize}
        \item  $\mathcal{X} =~ \{(x,y)\in\RR^2 \mid x^2 + 4 y^2 \le 4\}$.
        \item $U =~ \{(x,y)\in\RR^2 \mid (x-0.5)^2 + (y - 0.75)^2 \le 0.05^2\}$.
        \item  $I ~=~  \{(x,y)\in\RR^2 \mid x^2 + (y+0.5)^2 \le 0.01^2\}$.
    \end{itemize}  
\end{example}

\smallskip

\begin{example}[\textnormal{\textsf{lyapunov}~\cite{RS10}}]
The system dynamic is: (safety)
    \begin{align*}
        x_{n+1} & \eeq x_n - 0.2\theta y_n\\
        y_{n+1} & \eeq y_n - 0.2\theta z_n\\
        z_{n+1} & \eeq z_n +0.1(-x_n - 2y_n - z_n +x_n^3)
    \end{align*}
    \begin{itemize}
        \item  $\mathcal{X} =~ \{(x,y)\in\RR^2 \mid x^2 +  y^2 + z^2 \le 4\}$.
        \item $U =~ \{(x,y)\in\RR^2 \mid (x-0.5)^2 + (y-0.5)^2 + (z - 0.5)^2\le 0.2^2\}$.
        \item  $I ~=~  \{(x,y)\in\RR^2 \mid (x-0.25)^2 + (y-0.25)^2 + (z-0.25)^2 \le 0.2^2\}$.
    \end{itemize}  
\end{example}

\smallskip

\begin{example}[\textnormal{\textsf{lotka}~\cite{GJP+14}}]
The system dynamic is: (safety)
    \begin{align*}
        x_{n+1} & \eeq x_n + 0.1x_n(\theta-z_n)\\
        y_{n+1} & \eeq y_n + 0.1y_n(1-2z_n)\\
        z_{n+1} & \eeq z_n +0.1z_n(x_n+y_n-1)
    \end{align*}
    \begin{itemize}
        \item  $\mathcal{X} =~ \{(x,y)\in\RR^2 \mid x^2 +  y^2 + z^2 \le 1\}$.
        \item $U =~ \{(x,y)\in\RR^2 \mid x^2 + y^2 \le 0.2^2\}$.
        \item  $I ~=~  \{(x,y)\in\RR^2 \mid (x-0.5)^2 + (y-0.5)^2 + z^2 \le 0.4^2\}$.
    \end{itemize}  
\end{example}

%% file: aaai2026.bib
@inproceedings{LZZ11,
	title        = {Computing semi-algebraic invariants for polynomial dynamical systems},
	author       = {Liu, Jiang and Zhan, Naijun and Zhao, Hengjun},
	year         = 2011,
	booktitle    = {{EMSOFT}'11},
	publisher    = {ACM},
	pages        = {97--106},
	_address     = {New York, NY, USA}
}

@article{APLS08,
	title        = {Probabilistic reachability and safety for controlled discrete time stochastic hybrid systems},
	author       = {A. Abate and M. Prandini and J. Lygeros and S. Sastry},
	year         = 2008,
	journal      = {Automatica},
	volume       = 44,
	number       = 11,
	pages        = {2724--2734}
}

@inproceedings{prajna2004stochastic,
	title        = {Stochastic safety verification using barrier certificates},
	author       = {Prajna, Stephen and Jadbabaie, Ali and Pappas, George J},
	year         = 2004,
	booktitle    = {{CDC}'04},
	volume       = 1,
	pages        = {929--934},
	organization = {IEEE}
}

@inproceedings{Lofberg2004,
	title        = {{YALMIP}: A toolbox for modeling and optimization in {MATLAB}},
	author       = {L{\"{o}}fberg, J.},
	year         = 2004,
	booktitle    = {{CACSD}'04},
	pages        = {284--289}
}

@article{DBLP:journals/mp/AndersenRT03,
	title        = {On implementing a primal-dual interior-point method for conic quadratic optimization},
	author       = {Erling D. Andersen and Cornelis Roos and Tam{\'{a}}s Terlaky},
	year         = 2003,
	journal      = {Math. Program.},
	volume       = 95,
	number       = 2,
	pages        = {249--277},
	_url         = {https://doi.org/10.1007/s10107-002-0349-3},
	_doi         = {10.1007/s10107-002-0349-3}
}

@inproceedings{Sogokon+Others/2018/Vector,
	title        = {Vector Barrier Certificates and Comparison Systems},
	author       = {Andrew Sogokon and Khalil Ghorbal and Yong Kiam Tan and Andr{\'{e}} Platzer},
	year         = 2018,
	booktitle    = {{FM}'18},
	publisher    = {Springer},
	series       = {Lecture Notes in Computer Science},
	volume       = 10951,
	pages        = {418--437}
}

@book{williams1991probability,
	title={Probability with martingales},
	author={Williams, David},
	year={1991},
	publisher={Cambridge university press}
}

@book{durrett2019probability,
	title={Probability: theory and examples},
	author={Durrett, Rick},
	volume={49},
	year={2019},
	publisher={Cambridge university press}
}

@inproceedings{feng2020unbounded,
	title={Unbounded-time safety verification of stochastic differential dynamics},
	author={Feng, Shenghua and Chen, Mingshuai and Xue, Bai and Sankaranarayanan, Sriram and Zhan, Naijun},
	booktitle={International Conference on Computer Aided Verification},
	pages={327--348},
	year={2020},
	organization={Springer}
}

@article{ames2016control,
	title={Control barrier function based quadratic programs for safety critical systems},
	author={Ames, Aaron D and Xu, Xiangru and Grizzle, Jessy W and Tabuada, Paulo},
	journal={IEEE Transactions on Automatic Control},
	volume={62},
	number={8},
	pages={3861--3876},
	year={2016},
	publisher={IEEE}
}

@article{jagtap2020formal,
	title={Formal synthesis of stochastic systems via control barrier certificates},
	author={Jagtap, Pushpak and Soudjani, Sadegh and Zamani, Majid},
	journal={IEEE Transactions on Automatic Control},
	volume={66},
	number={7},
	pages={3097--3110},
	year={2020},
	publisher={IEEE}
}

@inproceedings{chakarov2013probabilistic,
  title={Probabilistic program analysis with martingales},
  author={Chakarov, Aleksandar and Sankaranarayanan, Sriram},
  booktitle={International Conference on Computer Aided Verification},
  pages={511--526},
  year={2013},
  organization={Springer}
}

@inproceedings{chatterjee2016termination,
  title={Termination analysis of probabilistic programs through Positivstellensatz’s},
  author={Chatterjee, Krishnendu and Fu, Hongfei and Goharshady, Amir Kafshdar},
  booktitle={International Conference on Computer Aided Verification},
  pages={3--22},
  year={2016},
  organization={Springer}
}

@article{majumdar2025sound,
  title={Sound and complete proof rules for probabilistic termination},
  author={Majumdar, Rupak and Sathiyanarayana, VR},
  journal={Proceedings of the ACM on Programming Languages},
  volume={9},
  number={POPL},
  pages={1871--1902},
  year={2025},
  publisher={ACM New York, NY, USA}
}

@article{kanamaru2007van,
  title={Van der Pol oscillator},
  author={Kanamaru, Takashi},
  journal={Scholarpedia},
  volume={2},
  number={1},
  pages={2202},
  year={2007}
}

@article{prajna2007framework,
  title={A framework for worst-case and stochastic safety verification using barrier certificates},
  author={Prajna, Stephen and Jadbabaie, Ali and Pappas, George J},
  journal={IEEE Transactions on Automatic Control},
  volume={52},
  number={8},
  pages={1415--1428},
  year={2007},
  publisher={IEEE}
}

@book{bertsekas1996stochastic,
	title={Stochastic optimal control: {T}he discrete-time case},
	author={Bertsekas, Dimitri and Shreve, Steven E},
	volume={5},
	year={1996},
	publisher={Athena Scientific}
}

@book{paul2013stochastic,
	title={Stochastic Processes: {F}rom Physics to Finance},
	author={Paul, W. and Baschnagel, J.},
	_isbn={9783319003276},
	year={2013},
	publisher={Springer}
}

@book{steele2001stochastic,
	title={Stochastic calculus and financial applications},
	author={Steele, J Michael},
	volume={1},
	year={2001},
	publisher={Springer}
}

@book{allen2010introduction,
	title={An introduction to stochastic processes with applications to biology},
	author={Allen, Linda JS},
	year={2010},
	publisher={CRC press}
}

@inproceedings{lechner2022stability,
  title={Stability verification in stochastic control systems via neural network supermartingales},
  author={Lechner, Mathias and {{Z}}ikeli{c}, {DJ}or{dj}e and Chatterjee, Krishnendu and Henzinger, Thomas A},
  booktitle={Proceedings of the aaai conference on artificial intelligence},
  volume={36},
  number={7},
  pages={7326--7336},
  year={2022}
}

@inproceedings{vzikelic2023learning,
  title={Learning control policies for stochastic systems with reach-avoid guarantees},
  author={{\v{Z}}ikeli{\'c}, {D}or{d}e and Lechner, Mathias and Henzinger, Thomas A and Chatterjee, Krishnendu},
  booktitle={Proceedings of the AAAI Conference on Artificial Intelligence},
  volume={37},
  number={10},
  pages={11926--11935},
  year={2023}
}

@inproceedings{DBLP:conf/qest/HenriquesMZPC12,
	author    = {David Henriques and
	Jo{\~{a}}o G. Martins and
	Paolo Zuliani and
	Andr{\'{e}} Platzer and
	Edmund M. Clarke},
	title     = {Statistical Model Checking for {M}arkov Decision Processes},
	booktitle = {{QEST}},
	pages     = {84--93},
	publisher = {{IEEE} Computer Society},
	year      = {2012}
}

@inproceedings{DBLP:conf/sefm/LegayST14,
	author    = {Axel Legay and
	Sean Sedwards and
	Louis{-}Marie Traonouez},
	title     = {Scalable Verification of {M}arkov Decision Processes},
	booktitle = {{SEFM} Workshops},
	series    = {LNCS},
	volume    = {8938},
	pages     = {350--362},
	publisher = {Springer},
	year      = {2014}
}

@article{DBLP:journals/automatica/SummersL10,
	author    = {Sean Summers and
	John Lygeros},
	title     = {Verification of discrete time stochastic hybrid systems: {A} stochastic
	reach-avoid decision problem},
	journal   = {Autom.},
	volume    = {46},
	number    = {12},
	pages     = {1951--1961},
	year      = {2010}
}

@article{putinar1993positive,
	title={Positive polynomials on compact semi-algebraic sets},
	author={Putinar, Mihai},
	journal={Indiana University Mathematics Journal},
	volume={42},
	number={3},
	pages={969--984},
	year={1993},
	publisher={JSTOR}
}

@inproceedings{DBLP:conf/lics/Kwiatkowska03,
	author    = {Marta Z. Kwiatkowska},
	title     = {Model Checking for Probability and Time: {F}rom Theory to Practice},
	booktitle = {{LICS}},
	pages     = {351},
	publisher = {{IEEE} Computer Society},
	year      = {2003},
	_url       = {https://doi.org/10.1109/LICS.2003.1210075},
	_doi       = {10.1109/LICS.2003.1210075}
}

@article{vandenberghe1996semidefinite,
	title={Semidefinite programming},
	author={Vandenberghe, Lieven and Boyd, Stephen},
	journal={SIAM review},
	volume={38},
	number={1},
	pages={49--95},
	year={1996},
	publisher={SIAM}
}

@article{parrilo2003semidefinite,
	title={Semidefinite programming relaxations for semialgebraic problems},
	author={Parrilo, Pablo A},
	journal={Mathematical programming},
	volume={96},
	number={2},
	pages={293--320},
	year={2003},
	publisher={Springer}
}

@inproceedings{SGJ16,
  author    = {Andrew Sogokon and
               Khalil Ghorbal and
               Taylor T. Johnson},
  editor    = {Goran Frehse and
               Matthias Althoff},
  title     = {Non-linear Continuous Systems for Safety Verification},
  booktitle = {ARCH@CPSWeek},
  series    = {EPiC Series in Computing},
  volume    = {43},
  pages     = {42--51},
  year      = {2016}
}

@article{XZF22,
  title={Reach-Avoid Analysis for Stochastic Differential Equations},
  author={Xue, Bai and Zhan, Naijun and Fr{\"a}nzle, Martin},
  journal={arXiv preprint arXiv:2208.10752},
  year={2022}
}

@article{PJP07,
  author    = {Stephen Prajna and
               Ali Jadbabaie and
               George J. Pappas},
  title     = {A Framework for Worst-Case and Stochastic Safety Verification Using
               Barrier Certificates},
  journal   = {{IEEE} Trans. Autom. Control.},
  volume    = {52},
  number    = {8},
  pages     = {1415--1428},
  year      = {2007}
}

@inproceedings{GJP+14,
  author    = {Eric Goubault and
               Jacques{-}Henri Jourdan and
               Sylvie Putot and
               Sriram Sankaranarayanan},
  title     = {Finding non-polynomial positive invariants and lyapunov functions
               for polynomial systems through Darboux polynomials},
  booktitle = {{ACC}},
  pages     = {3571--3578},
  publisher = {{IEEE}},
  year      = {2014}
}

@article{RS10,
  author    = {Stefan Ratschan and
               Zhikun She},
  title     = {Providing a Basin of Attraction to a Target Region of Polynomial Systems by Computation of Lyapunov-Like Functions},
  journal   = {{SIAM} J. Control. Optim.},
  volume    = {48},
  number    = {7},
  pages     = {4377--4394},
  year      = {2010},
}

@inproceedings{Kong13,
	author    = {Hui Kong and Fei He and Xiaoyu Song and William N. N. Hung and Ming Gu},
	title     = {Exponential-condition-based barrier certificate generation for safety verification of hybrid systems},
	booktitle = {{CAV}},
	year      = {2013},
	volume    = {8044},
	series    = {LNCS},
	pages     = {242--257},
	publisher = {Springer},
}

@inproceedings{zeng2016darboux,
	title={{D}arboux-type barrier certificates for safety verification of nonlinear hybrid systems},
	author={Zeng, Xia and Lin, Wang and Yang, Zhengfeng and Chen, Xin and Wang, Lilei},
	booktitle={EMSOFT},
	pages={1--10},
	year={2016},
	publisher = {{ACM}}
}

@article{Gan17,
	author    = {Liyun Dai and
	Ting Gan and
	Bican Xia and
	Naijun Zhan},
	title     = {Barrier certificates revisited},
	journal   = {J. Symb. Comput.},
	volume    = {80},
	pages     = {62--86},
	year      = {2017}
}

@InProceedings{xu2015robustness,
	author    = {Xiangru Xu and Paulo Tabuada and Jessy W. Grizzle and Aaron D. Ames},
	title     = {Robustness of control barrier functions for safety critical control},
	booktitle = {{ADHS}},
	year      = {2015},
	volume    = {48},
	number    = {27},
	_series    = {IFAC-PapersOnLine},
	pages     = {54--61},
	publisher = {Elsevier},
	_doi      = {10.1016/j.ifacol.2015.11.152},
	_url      = {https://doi.org/10.1016/j.ifacol.2015.11.152},
}

@Article{wang2017generating,
	author    = {Wang, Qiuye and Li, Yangjia and Xia, Bican and Zhan, Naijun},
	title     = {Generating semi-algebraic invariants for non-autonomous polynomial hybrid systems},
	journal   = {J. Syst. Sci. Complex.},
	year      = {2017},
	volume    = {30},
	number    = {1},
	pages     = {234--252},
	publisher = {Springer},
}

@inproceedings{DBLP:conf/cav/FengC00Z20,
	author    = {Shenghua Feng and
	Mingshuai Chen and
	Bai Xue and
	Sriram Sankaranarayanan and
	Naijun Zhan},
	title     = {Unbounded-Time Safety Verification of Stochastic Differential Dynamics},
	booktitle = {{CAV} {(II)}},
	series    = {LNCS},
	volume    = {12225},
	pages     = {327--348},
	publisher = {Springer},
	year      = {2020}
}

@article{DBLP:journals/tac/LahijanianAB15,
	author    = {Morteza Lahijanian and
	Sean B. Andersson and
	Calin Belta},
	title     = {Formal Verification and Synthesis for Discrete-Time Stochastic Systems},
	journal   = {{IEEE} Trans. Autom. Control.},
	volume    = {60},
	number    = {8},
	pages     = {2031--2045},
	year      = {2015}
}

@book{baier2008principles,
	title={Principles of Model Checking},
	author={Baier, Christel and Katoen, Joost-Pieter},
	year={2008},
	publisher={MIT press}
}

@inproceedings{DBLP:conf/cav/Baier0L0W17,
	author    = {Christel Baier and
	Joachim Klein and
	Linda Leuschner and
	David Parker and
	Sascha Wunderlich},
	title     = {Ensuring the Reliability of Your Model Checker: {I}nterval Iteration
	for {M}arkov Decision Processes},
	booktitle = {{CAV} {(II)}},
	series    = {{LNCS}},
	volume    = {10426},
	pages     = {160--180},
	publisher = {Springer},
	year      = {2017},
	_url       = {https://doi.org/10.1007/978-3-319-63387-9\_8},
	_doi       = {10.1007/978-3-319-63387-9\_8}
}

@inproceedings{DBLP:conf/cav/QuatmannK18,
	author    = {Tim Quatmann and
	Joost{-}Pieter Katoen},
	title     = {Sound Value Iteration},
	booktitle = {{CAV} {(I)}},
	series    = {{LNCS}},
	volume    = {10981},
	pages     = {643--661},
	publisher = {Springer},
	year      = {2018},
	_url       = {https://doi.org/10.1007/978-3-319-96145-3\_37},
	_doi       = {10.1007/978-3-319-96145-3\_37}
}

@inproceedings{DBLP:conf/cav/HartmannsK20,
	author    = {Arnd Hartmanns and
	Benjamin Lucien Kaminski},
	title     = {Optimistic Value Iteration},
	booktitle = {{CAV} {(II)}},
	series    = {{LNCS}},
	volume    = {12225},
	pages     = {488--511},
	publisher = {Springer},
	year      = {2020},
	_url       = {https://doi.org/10.1007/978-3-030-53291-8\_26},
	_doi       = {10.1007/978-3-030-53291-8\_26}
}

@incollection{vardi2005automata,
  title={An automata-theoretic approach to linear temporal logic},
  author={Vardi, Moshe Y},
  booktitle={Logics for concurrency: structure versus automata},
  pages={238--266},
  year={2005},
  publisher={Springer}
}

@incollection{bartocci2018introduction,
  title={Introduction to runtime verification},
  author={Bartocci, Ezio and Falcone, Yli{\`e}s and Francalanza, Adrian and Reger, Giles},
  booktitle={Lectures on Runtime Verification: Introductory and Advanced Topics},
  pages={1--33},
  year={2018},
  publisher={Springer}
}

@article{deshmukh2017robust,
  title={Robust online monitoring of signal temporal logic},
  author={Deshmukh, Jyotirmoy V and Donz{\'e}, Alexandre and Ghosh, Shromona and Jin, Xiaoqing and Juniwal, Garvit and Seshia, Sanjit A},
  journal={Formal Methods in System Design},
  volume={51},
  number={1},
  pages={5--30},
  year={2017},
  publisher={Springer}
}

@inproceedings{raman2015reactive,
  title={Reactive synthesis from signal temporal logic specifications},
  author={Raman, Vasumathi and Donz{\'e}, Alexandre and Sadigh, Dorsa and Murray, Richard M and Seshia, Sanjit A},
  booktitle={Proceedings of the 18th international conference on hybrid systems: Computation and control},
  pages={239--248},
  year={2015}
}

@inproceedings{su2025runtime,
  title={Runtime enforcement of CPS against signal temporal logic},
  author={Su, Han and Shankar, Saumya and Pinisetty, Srinivas and Roop, Partha S and Zhan, Naijun},
  booktitle={Proceedings of the 28th ACM International Conference on Hybrid Systems: Computation and Control},
  pages={1--11},
  year={2025}
}

@inproceedings{yu2025neural,
  title={Neural Control and Certificate Repair via Runtime Monitoring},
  author={Yu, Emily and {\v{Z}}ikeli{\'c}, {D}or{d}e and Henzinger, Thomas A},
  booktitle={Proceedings of the AAAI Conference on Artificial Intelligence},
  volume={39},
  number={25},
  pages={26409--26417},
  year={2025}
}

@article{dawson2023safe,
  title={Safe control with learned certificates: A survey of neural lyapunov, barrier, and contraction methods for robotics and control},
  author={Dawson, Charles and Gao, Sicun and Fan, Chuchu},
  journal={IEEE Transactions on Robotics},
  volume={39},
  number={3},
  pages={1749--1767},
  year={2023},
  publisher={IEEE}
}
